\documentclass[12pt,twoside]{article}
\usepackage{amsmath,amsfonts,amssymb,amsthm,txfonts,hyperref,url,graphicx}
\newtheorem{theorem}{Theorem}
\newtheorem{lemma}{Lemma}
\newcommand{\N}{\mathbb N}
\newcommand{\Z}{\mathbb Z}
\newcommand{\K}{\textbf{\textit{K}}}
\author{Apoloniusz Tyszka}
\title{\large {\sl MuPAD} codes which implement \mbox{limit-computable}
functions that cannot be bounded by any computable function}
\date{}
\begin{document}
\begin{sloppypar}
\def\UrlFont{\em}
\maketitle
\begin{abstract}
For a positive integer $n$, let \mbox{$f(n)$} denote the smallest \mbox{non-negative}
integer $b$ such that for each system
\mbox{$S \subseteq \{x_k=1,~x_i+x_j=x_k,~x_i \cdot x_j=x_k:~i,j,k \in \{1,\ldots,n\}\}$}
with a solution in \mbox{non-negative} integers \mbox{$x_1,\ldots,x_n$} there exists a solution
of $S$ in \mbox{non-negative} integers not greater than $b$. We prove that the function $f$
is strictly increasing and dominates all computable functions.
We present an infinite loop in {\sl MuPAD} which takes as input a positive integer $n$ and
returns a \mbox{non-negative} integer on each iteration. Let \mbox{$g(n,m)$} denote the number
returned on the \mbox{$m$-th} iteration, if $n$ is taken as input. Then,
\vskip 0.2truecm
\noindent
\centerline{$g(n,m) \leq m-1$,}
\vskip 0.2truecm
\noindent
\centerline{$0=g(n,1)<1=g(n,2) \leq g(n,3) \leq g(n,4) \leq \ldots$}
\vskip 0.01truecm
\noindent
and
\vskip 0.01truecm
\noindent
\centerline{$g(n,f(n))<f(n)=g(n,f(n)+1)=g(n,f(n)+2)=g(n,f(n)+3)=\ldots$}
\vskip 0.2truecm
\par
A {\sl MuPAD} code constructed on another principle contains a \mbox{repeat-until}
loop and implements a \mbox{limit-computable} function
\mbox{$\xi: \N \to \N$} that cannot be bounded by any computable function. This code takes as input
a \mbox{non-negative} integer $n$, immediately returns $0$, and computes a system $S$ of polynomial
equations. If the loop terminates for $S$, then the next instruction is executed and returns \mbox{$\xi(n)$}.
\end{abstract}
{\sl MuPAD}, \mbox{trial-and-error} computable function.
\newpage
{Limit-computable} functions, also known as \mbox{trial-and-error} computable functions,
have been thoroughly studied, see \mbox{\cite[pp.~233--235]{Soare}} for the main results.
Our first goal is to present an infinite loop in {\sl MuPAD} which finds the values
of a \mbox{limit-computable} function \mbox{$f:\N \setminus \{0\} \to \N \setminus \{0\}$}
by an infinite computation, where $f$ dominates all computable functions.
There are many \mbox{limit-computable} functions \mbox{$f:\N \setminus \{0\} \to \N \setminus \{0\}$}
which cannot be bounded by any computable function. For example, this follows from
\cite[p.~38,~item~4]{Janiczak}, see also \mbox{\cite[p.~268]{Murawski}} where Janiczak's result is mentioned.
Unfortunately, for all known such functions~$f$, it is difficult to write a suitable computer program.
\mbox{A sophisticated} choice of a \mbox{function $f$} will allow us to do it.
\vskip 0.2truecm
\par
Let
\vskip 0.2truecm
\noindent
\centerline{$E_n=\{x_k=1,~x_i+x_j=x_k,~x_i \cdot x_j=x_k:~i,j,k \in \{1,\ldots,n\}\}$}
\vskip 0.2truecm
\par
For a positive integer $n$, let \mbox{$f(n)$} denote the smallest \mbox{non-negative} integer~$b$ such that
for each system \mbox{$S \subseteq E_n$} with a solution in \mbox{non-negative} integers \mbox{$x_1,\ldots,x_n$}
there exists a solution of $S$ in \mbox{non-negative} integers not greater than~$b$.
This definition is correct because there are only finitely many subsets of \mbox{$E_n$}.
\vskip 0.2truecm
\par
The system
\begin{displaymath}
\left\{
\begin{array}{rcl}
x_1 &=& 1 \\
x_1+x_1 &=& x_2 \\
x_2 \cdot x_2 &=& x_3 \\
x_3 \cdot x_3 &=& x_4 \\
&\ldots& \\
x_{n-1} \cdot x_{n-1} &=&x_n
\end{array}
\right.
\end{displaymath}
\noindent
has a unique integer solution, namely
\mbox{$\left(1,2,4,16,\ldots,2^{\textstyle 2^{n-3}},2^{\textstyle 2^{n-2}}\right)$}.
Therefore, \mbox{$f(1)=1$} and \mbox{$f(n) \geq 2^{\textstyle 2^{n-2}} \geq 2$} for any \mbox{$n \geq 2$}.
\begin{lemma}\label{lem0}
For each integer \mbox{$n \geq 2$}, \mbox{$f(n+1) \geq f(n)^2>f(n)$}.
\end{lemma}
\begin{proof}
If a system \mbox{$S \subseteq E_n$} has a solution in \mbox{non-negative}
integers \mbox{$x_1,\ldots,x_n$}, then for each \mbox{$i \in \{1,\ldots,n\}$} the system
\mbox{$S \cup \{x_i \cdot x_i=x_{n+1}\} \subseteq E_{n+1}$} has a solution
in \mbox{non-negative} integers \mbox{$x_1,\ldots,x_{n+1}$}.
\end{proof}
\newpage
For each integer \mbox{$n \geq 12$}, the following system \mbox{$S \subseteq E_n$}
\vskip 0.2truecm
\noindent
\centerline{$x_1=1$~~~~~~~~~~$x_1+x_1=x_2$~~~~~~~~~~$x_2+x_2=x_3$~~~~~~~~~~$x_1+x_3=x_4$}
\vskip 0.2truecm
\noindent
\centerline{$x_4 \cdot x_4=x_5$~~~~~~~~~~$x_5 \cdot x_5=x_6$~~~~~~~~~~$x_6 \cdot x_7=x_8$~~~~~~~~~~$x_8 \cdot x_8=x_9$}
\vskip 0.2truecm
\noindent
\centerline{$x_{10} \cdot x_{10}=x_{11}$~~~~~~~~~~$x_{11}+x_1=x_{12}$~~~~~~~~~~$x_4 \cdot x_9=x_{12}$}
\vskip 0.2truecm
\noindent
\centerline{$x_{12} \cdot x_{12}=x_{13}$~~~~~~~~~~$x_{13} \cdot x_{13}=x_{14}$~~~~~~~~~~\ldots~~~~~~~~~~$x_{n-1} \cdot x_{n-1}=x_n$}
\vskip 0.2truecm
\noindent
has infinitely many solutions in \mbox{non-negative} integers \mbox{$x_1,\ldots,x_n$} and each such solution \mbox{$(x_1,\ldots,x_n)$}
satisfies \mbox{${\rm max}(x_1,\ldots,x_n)>2^{\textstyle 2^{n-2}}$}, see \mbox{\cite[p.~4,~Theorem~1]{Tyszka9}}.
Hence, \mbox{$f(n)>2^{\textstyle 2^{n-2}}$} for any \mbox{$n \geq 12$}.
\vskip 0.2truecm
\par
The Davis-Putnam-Robinson-Matiyasevich theorem states that every recursively enumerable
set \mbox{${\cal M} \subseteq {\N}^n$} has a Diophantine representation, that is
\begin{equation}
\tag*{\tt (R)}
(a_1,\ldots,a_n) \in {\cal M} \Longleftrightarrow
\exists x_1, \ldots, x_m \in \N ~~W(a_1,\ldots,a_n,x_1,\ldots,x_m)=0
\end{equation}
for some polynomial $W$ with integer coefficients, see \cite{Matiyasevich}.
The polynomial~$W$ can be computed, if we know a Turing machine~$M$
such that, for all \mbox{$(a_1,\ldots,a_n) \in {\N}^n$}, $M$ halts on \mbox{$(a_1,\ldots,a_n)$} if and only if
\mbox{$(a_1,\ldots,a_n) \in {\cal M}$}, see \cite{Matiyasevich}.
The representation~{\tt (R)} is said to be \mbox{single-fold}, if for any \mbox{$a_1,\ldots,a_n \in \N$}
the equation \mbox{$W(a_1,\ldots,a_n,x_1,\ldots,x_m)=0$} has at most one solution \mbox{$(x_1,\ldots,x_m) \in {\N}^m$}.
\vskip 0.2truecm
\par
Let \mbox{$\cal{R}${\sl ng}} denote the class of all rings $\K$ that extend $\Z$.
\begin{lemma}\label{lem2}
(\mbox{\cite[p.~720]{Tyszka7}}) Let \mbox{$D(x_1,\ldots,x_p) \in {\Z}[x_1,\ldots,x_p]$}.
Assume that \mbox{$d_i={\rm deg}(D,x_i) \geq 1$} for each \mbox{$i \in \{1,\ldots,p\}$}. We can compute a positive
integer \mbox{$n>p$} and a system \mbox{$T \subseteq E_n$} which satisfies the following two conditions:
\vskip 0.2truecm
\noindent
{\tt Condition 1.} If \mbox{$\K \in {\cal R}{\sl ng} \cup \{\N,~\N \setminus \{0\}\}$}, then
\[
\forall \tilde{x}_1,\ldots,\tilde{x}_p \in \K ~\Bigl(D(\tilde{x}_1,\ldots,\tilde{x}_p)=0 \Longleftrightarrow
\]
\[
\exists \tilde{x}_{p+1},\ldots,\tilde{x}_n \in \K ~(\tilde{x}_1,\ldots,\tilde{x}_p,\tilde{x}_{p+1},\ldots,\tilde{x}_n) ~solves~ T\Bigr)
\]
{\tt Condition 2.} If \mbox{$\K \in {\cal R}{\sl ng} \cup \{\N,~\N \setminus \{0\}\}$}, then
for each \mbox{$\tilde{x}_1,\ldots,\tilde{x}_p \in \K$} with \mbox{$D(\tilde{x}_1,\ldots,\tilde{x}_p)=0$},
there exists a unique tuple \mbox{$(\tilde{x}_{p+1},\ldots,\tilde{x}_n) \in {\K}^{n-p}$} such that the tuple
\mbox{$(\tilde{x}_1,\ldots,\tilde{x}_p,\tilde{x}_{p+1},\ldots,\tilde{x}_n)$} \mbox{solves $T$}.
\vskip 0.2truecm
\noindent
Conditions 1 and 2 imply that for each \mbox{$\K \in {\cal R}{\sl ng} \cup \{\N,~\N \setminus \{0\}\}$}, the equation
\mbox{$D(x_1,\ldots,x_p)=0$} and the \mbox{system $T$} have the same number of solutions \mbox{in $\K$}.
\end{lemma}
\begin{theorem}\label{the1} (\cite{Tyszka3a})
If a function \mbox{$\Gamma:\N \setminus \{0\} \to \N$} is computable, then there exists a positive
integer $m$ such that \mbox{$\Gamma(n)<f(n)$} for any \mbox{$n \geq m$}.
\end{theorem}
\begin{proof}
The Davis-Putnam-Robinson-Matiyasevich theorem and Lemma~\ref{lem2} for \mbox{$\K=\N$} imply that there
exists an integer \mbox{$s \geq 3$} such that for any \mbox{non-negative} integers~\mbox{$x_1,x_2$},
\begin{equation}
\tag*{\tt (E)}
(x_1,x_2) \in \Gamma \Longleftrightarrow \exists x_3,\ldots,x_s \in \N ~~\Phi(x_1,x_2,x_3,\ldots,x_s),
\end{equation}
where the formula \mbox{$\Phi(x_1,x_2,x_3,\ldots,x_s)$} is a conjunction of formulae of the forms
\mbox{$x_k=1$}, \mbox{$x_i+x_j=x_k$}, \mbox{$x_i \cdot x_j=x_k$} \mbox{$(i,j,k \in \{1,\ldots,s\})$}.
Let $[\cdot]$ denote the integer part function. For each integer \mbox{$n \geq 6+2s$},
\[
n-\left[\frac{n}{2}\right]-3-s \geq 6+2s-\left[\frac{6+2s}{2}\right]-3-s \geq 6+2s-\frac{6+2s}{2}-3-s=0
\]
For an integer \mbox{$n \geq 6+2s$}, let $S_n$ denote the following system
\[\left\{
\begin{array}{rcl}
{\rm all~equations~occurring~in~}\Phi(x_1,x_2,x_3,\ldots,x_s) \\
n-\left[\frac{n}{2}\right]-3-s {\rm ~equations~of~the~form~} z_i=1 \\
t_1 &=& 1 \\
t_1+t_1 &=& t_2 \\
t_2+t_1 &=& t_3 \\
&\ldots& \\
t_{\left[\frac{n}{2}\right]-1}+t_1 &=& t_{\left[\frac{n}{2}\right]} \\
t_{\left[\frac{n}{2}\right]}+t_{\left[\frac{n}{2}\right]} &=& w \\
w+y &=& x_1 \\
y+y &=& y {\rm ~(if~}n{\rm ~is~even)} \\
y &=& 1 {\rm ~(if~}n{\rm ~is~odd)} \\
x_2+t_1 &=& u 
\end{array}
\right.\]
with $n$ variables. By the equivalence~{\tt (E)}, $S_n$ is satisfiable over~$\N$.
If a \mbox{$n$-tuple} \mbox{$(x_1,x_2,x_3,\ldots,x_s,\ldots,w,y,u)$} of \mbox{non-negative}
integers solves $S_n$, then by the equivalence {\tt (E)},
\[
x_2=\Gamma(x_1)=\Gamma(w+y)=\Gamma\left(2 \cdot \left[\frac{n}{2}\right]+y\right)=\Gamma(n)
\]
Therefore, \mbox{$u=x_2+t_1=\Gamma(n)+1>\Gamma(n)$}. This shows that \mbox{$\Gamma(n)<f(n)$} for any \mbox{$n \geq 6+2s$}.
\end{proof}
\par
The following infinite loop in {\sl MuPAD} is also stored in \cite{Tyszka5}. It takes as input a positive
integer $n$ and returns a \mbox{non-negative} integer on each iteration. Unfortunately, on each iteration
the program executes a brute force algorithm, which is very time consuming.
\begin{verbatim}
input("input the value of n",n):
X:=[0]:
while TRUE do
Y:=combinat::cartesianProduct(X $i=1..n):
W:=combinat::cartesianProduct(X $i=1..n):
for s from 1 to nops(Y) do
for t from 1 to nops(Y) do
m:=0:
for i from 1 to n do
if Y[s][i]=1 and Y[t][i]<>1 then m:=1 end_if:
for j from i to n do
for k from 1 to n do
if Y[s][i]+Y[s][j]=Y[s][k] and Y[t][i]+Y[t][j]<>Y[t][k]
then m:=1 end_if:
if Y[s][i]*Y[s][j]=Y[s][k] and Y[t][i]*Y[t][j]<>Y[t][k]
then m:=1 end_if:
end_for:
end_for:
end_for:
if m=0 and max(Y[t][i] $i=1..n)<max(Y[s][i] $i=1..n)
then W:=listlib::setDifference(W,[Y[s]]) end_if:
end_for:
end_for:
print(max(max(W[z][u] $u=1..n) $z=1..nops(W))):
X:=append(X,nops(X)):
end_while:
\end{verbatim}
\newpage
\begin{theorem}\label{the2}
Let \mbox{$g(n,m)$} denote the number returned on the \mbox{$m$-th} iteration,
if $n$ is taken as input. Then,
\vskip 0.2truecm
\noindent
\centerline{$g(n,m) \leq m-1$,}
\vskip 0.2truecm
\noindent
\centerline{$0=g(n,1)<1=g(n,2) \leq g(n,3) \leq g(n,4) \leq \ldots$}
\vskip 0.2truecm
\noindent
and
\vskip 0.2truecm
\noindent
\centerline{$g(n,f(n))<f(n)=g(n,f(n)+1)=g(n,f(n)+2)=g(n,f(n)+3)=\ldots$}
\end{theorem}
\begin{proof}
Let us say that a tuple \mbox{$y=(y_1,\ldots,y_n) \in {\N}^n$} is a {\em duplicate}
of a tuple \mbox{$x=(x_1,\ldots,x_n) \in {\N}^n$}, if
\vskip 0.2truecm
\noindent
\centerline{$(\forall k \in \{1,\ldots,n\} ~(x_k=1 \Longrightarrow y_k=1)) ~\wedge$}
\par
\noindent
\centerline{$(\forall i,j,k \in \{1,\ldots,n\} ~(x_i+x_j=x_k \Longrightarrow y_i+y_j=y_k)) ~\wedge$}
\par
\noindent
\centerline{$(\forall i,j,k \in \{1,\ldots,n\} ~(x_i \cdot x_j=x_k \Longrightarrow y_i \cdot y_j=y_k))$}
\vskip 0.2truecm
\noindent
For each positive integer $n$, \mbox{$f(n)$} equals the smallest
\mbox{non-negative} integer $b$ such that for each \mbox{$x \in {\N}^n$}
there exists a duplicate of $x$ in \mbox{$\{0,\ldots,b\}^n$}.
For each positive integers $n$ and $m$, \mbox{$g(n,m)$} equals the smallest
\mbox{non-negative} integer $b$ such that for each \mbox{$x \in \{0,\ldots,m-1\}^n$}
there exists a duplicate of $x$ in \mbox{$\{0,\ldots,b\}^n$}.
\end{proof}
\par
If we replace the instruction:
\begin{verbatim}
         print(max(max(W[z][u] $u=1..n) $z=1..nops(W))):
\end{verbatim}
by the following two instructions:
\begin{verbatim}
         g:=max(max(W[z][u] $u=1..n) $z=1..nops(W)):
         if g=nops(X)-1 then print(g) end_if:
\end{verbatim}
then the changed code (\cite{Tyszka6}) performs an infinite computation of \mbox{$f(n)$}
which returns a finite sequence of \mbox{non-negative} integers. The same sequence is
returned by the algorithm presented in the flowchart below.
\newpage
\begin{center}
\includegraphics[scale=0.88]{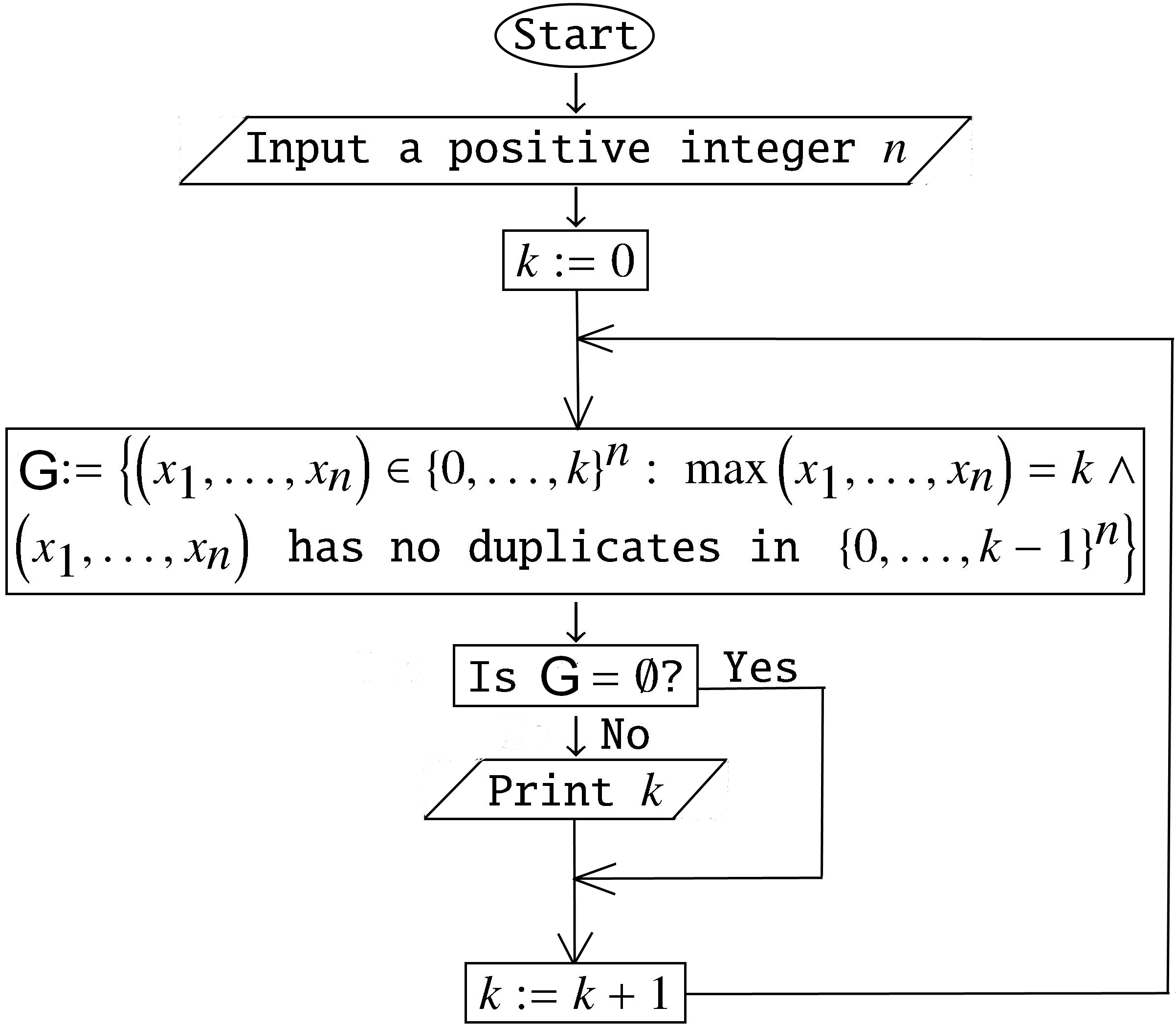}
\end{center}
\par
The next theorem is a corollary from Theorems~\ref{the1} and \ref{the2}.
\begin{theorem}\label{the3}
There exists a computable function \mbox{$\varphi: \N \times \N \to \N$} which satisfies the following conditions:
\vskip 0.2truecm
\noindent
{\tt 1)} For each \mbox{non-negative} integers $n$ and $l$,
\vskip 0.2truecm
\noindent
\centerline{$\varphi(n,l) \leq l$}
\vskip 0.2truecm
\noindent
{\tt 2)} For each \mbox{non-negative} integer $n$,
\vskip 0.2truecm
\noindent
\centerline{$0=\varphi(n,0)<1=\varphi(n,1) \leq \varphi(n,2) \leq \varphi(n,3) \leq \ldots$}
\vskip 0.2truecm
\noindent
{\tt 3)} For each \mbox{non-negative} integer $n$, the sequence \mbox{$\{\varphi(n,l)\}_{l \in \N}$}
is bounded from above.
\vskip 0.2truecm
\noindent
{\tt 4)} The function
\[
\N \ni n \stackrel{\normalsize \theta}{\normalsize \longrightarrow} \theta(n)=\lim_{l \to \infty} \varphi(n,l) \in \N \setminus \{0\}
\]
is strictly increasing and dominates all computable functions.
\vskip 0.2truecm
\noindent
{\tt 5)} For each \mbox{non-negative} integer $n$,
\vskip 0.2truecm
\noindent
\centerline{$\varphi(n,\theta(n)-1)<\theta(n)=\varphi(n,\theta(n))=\varphi(n,\theta(n)+1)=\varphi(n,\theta(n)+2)=\ldots$}
\end{theorem}
\begin{proof}
Let \mbox{$\varphi(n,l)=g(n+1,l+1)$}. The following {\sl MuPAD} code,
which is also stored in \cite{Tyszka3}, computes the values of \mbox{$\varphi(n,l)$}.
\begin{verbatim}
input("input the value of n",n):
input("input the value of l",l):
n:=n+1:
X:=[i $ i=0..l]:
Y:=combinat::cartesianProduct(X $i=1..n):
W:=combinat::cartesianProduct(X $i=1..n):
for s from 1 to nops(Y) do
for t from 1 to nops(Y) do
m:=0:
for i from 1 to n do
if Y[s][i]=1 and Y[t][i]<>1 then m:=1 end_if:
for j from i to n do
for k from 1 to n do
if Y[s][i]+Y[s][j]=Y[s][k] and Y[t][i]+Y[t][j]<>Y[t][k]
then m:=1 end_if:
if Y[s][i]*Y[s][j]=Y[s][k] and Y[t][i]*Y[t][j]<>Y[t][k]
then m:=1 end_if:
end_for:
end_for:
end_for:
if m=0 and max(Y[t][i] $i=1..n)<max(Y[s][i] $i=1..n)
then W:=listlib::setDifference(W,[Y[s]]) end_if:
end_for:
end_for:
print(max(max(W[z][u] $u=1..n) $z=1..nops(W))):
\end{verbatim}
\end{proof}
\par
Let us fix a computable enumeration \mbox{$D_0,D_1,D_2,\ldots$} of all Diophantine equations.
The following flowchart illustrates an infinite computation of a \mbox{limit-computable} function
that cannot be bounded by any computable function. For each \mbox{non-negative} integer $n$,
the function has a \mbox{non-zero} value at $n$ if and only if the equation $D_n$ has a solution
in \mbox{non-negative} integers. Unfortunately, the function does not have any easy implementation.
\newpage
\begin{center}
\includegraphics[scale=0.88]{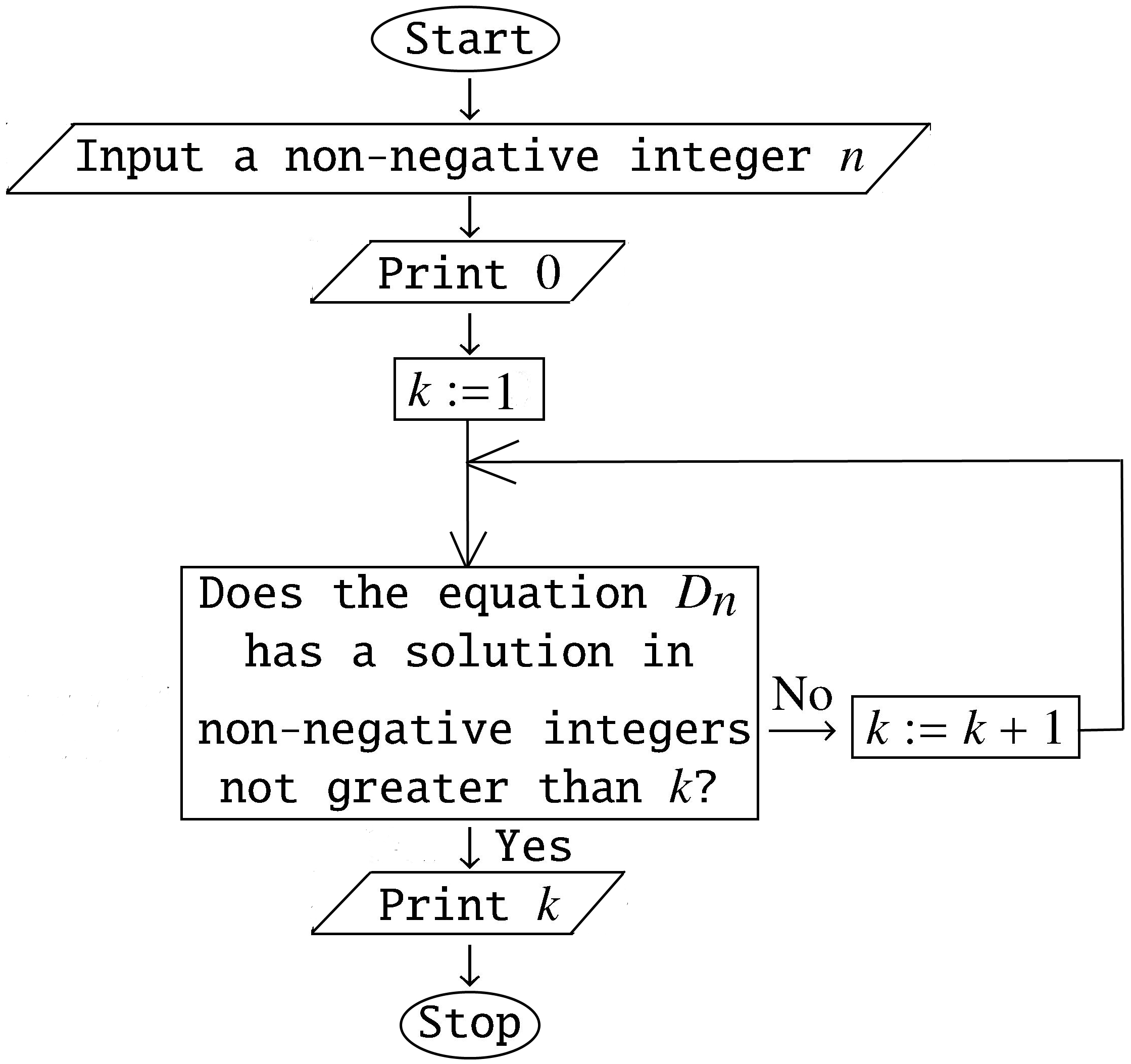}
\end{center}
\par
Frequently, it is hardly to decide whether or not a \mbox{limit-computable} function is computable.
For a positive integer $n$, let \mbox{$\chi(n)$} denote the smallest \mbox{non-negative} integer $b$ such
that for each system \mbox{$S \subseteq E_n$} with a unique solution in \mbox{non-negative} integers
\mbox{$x_1,\ldots,x_n$} this solution belongs to \mbox{$[0,b]^n$}. The function $\chi$ is
\mbox{limit-computable} and the code presented in \cite{Tyszka1a} (and stored in~\cite{Tyszkasf}) performs
an infinite computation of $\chi$. A computability of $\chi$ is an open problem.
Matiyasevich's conjecture on \mbox{single-fold} Diophantine representations~(\cite{Matiyasevich1})
implies that $\chi$ is not computable. In this case, $\chi$ dominates all computable functions.
\vskip 0.2truecm
\par
The following {\sl MuPAD} code is also stored in \cite{Tyszka1}.
\begin{verbatim}
input("input the value of n",n):
print(0):
A:=op(ifactor(210*(n+1))):
B:=[A[2*i+1] $i=1..(nops(A)-1)/2]:
S:={}:
for i from 1 to floor(nops(B)/4) do
if B[4*i]=1 then S:=S union {B[4*i-3]} end_if:
if B[4*i]=2 then S:=S union {[B[4*i-3],B[4*i-2],B[4*i-1],"+"]} end_if:
if B[4*i]>2 then S:=S union {[B[4*i-3],B[4*i-2],B[4*i-1],"*"]} end_if:
end_for:
m:=2:
repeat
C:=op(ifactor(m)):
W:=[C[2*i+1]-1 $i=1..(nops(C)-1)/2]:
T:={}:
for i from 1 to nops(W) do
for j from 1 to nops(W) do
for k from 1 to nops(W) do
if W[i]=1 then T:=T union {i} end_if:
if W[i]+W[j]=W[k] then T:=T union {[i,j,k,"+"]} end_if:
if W[i]*W[j]=W[k] then T:=T union {[i,j,k,"*"]} end_if:
end_for:
end_for:
end_for:
m:=m+1:
until S minus T={} end_repeat:
print(max(W[i] $i=1..nops(W))):
\end{verbatim}
\begin{theorem}\label{the4}
The above code implements a \mbox{limit-computable} function \mbox{$\xi: \N \to \N$} that cannot
be bounded by any computable function. The code takes as input a \mbox{non-negative} integer $n$,
immediately returns $0$, and computes a system $S$ of polynomial equations. If the \mbox{repeat-until}
loop terminates for $S$, then the next instruction is executed and returns \mbox{$\xi(n)$}.
\end{theorem}
\begin{proof}
Let \mbox{$n \in \N$}, and let \mbox{${p_1}^{t(1)} \cdot \ldots \cdot {p_s}^{t(s)}$} be a prime
factorization of \mbox{$210 \cdot (n+1)$}, where \mbox{$t(1),\ldots,t(s)$} denote positive
integers. Obviously, \mbox{$p_1=2$}, \mbox{$p_2=3$}, \mbox{$p_3=5$}, and \mbox{$p_4=7$}.
\vskip 0.2truecm
\par
For each positive integer $i$ that satifies \mbox{$4i \leq s$} and \mbox{$t(4i)=1$},
the code constructs the equation \mbox{$x_{t(4i-3)}=1$}.
\vskip 0.2truecm
\par
For each positive integer $i$ that satifies \mbox{$4i \leq s$} and \mbox{$t(4i)=2$},
the code constructs the equation \mbox{$x_{t(4i-3)}+x_{t(4i-2)}=x_{t(4i-1)}$}.
\vskip 0.2truecm
\par
For each positive integer $i$ that satifies \mbox{$4i \leq s$} and \mbox{$t(4i)>2$},
the code constructs the equation \mbox{$x_{t(4i-3)} \cdot x_{t(4i-2)}=x_{t(4i-1)}$}.
\vskip 0.2truecm
\par
The last three facts imply that the code assigns to $n$ a finite and \mbox{non-empty}
system $S$ which consists of equations of the forms: \mbox{$x_k=1$}, \mbox{$x_i+x_j=x_k$},
and \mbox{$x_i \cdot x_j=x_k$}. Conversely, each such system $S$ is assigned to some
\mbox{non-negative} integer $n$.
\vskip 0.2truecm
\par
Starting with the instruction $m:=2$, the code tries to find a solution of $S$
in \mbox{non-negative} integers by performing a \mbox{brute-force} search. If a solution
exists, then the searching terminates and the code returns a \mbox{non-negative} integer
\mbox{$\xi(n)$} such that the system $S$ has a solution in \mbox{non-negative} integers
not greater than \mbox{$\xi(n)$}. In the opposite case, execution of the code never terminates.
\vskip 0.2truecm
\par
A negative solution to Hilbert's Tenth Problem (\cite{Matiyasevich}) and Lemma~\ref{lem2} for \mbox{$\K=\N$}
imply that the code implements a \mbox{limit-computable} function \mbox{$\xi: \N \to \N$} that cannot
be bounded by any computable function.
\end{proof}
\par
Execution of the last code does not terminate for $n=7 \cdot 11 \cdot 13 \cdot 17 \cdot 19-1=323322$,
when the code tries to find a solution of the system \mbox{$\{x_1+x_1=x_1,~x_1=1\}$}. Execution terminates
for any \mbox{$n<323322$}, when the code returns $0$ and next $1$ \mbox{or $0$}. The last claim holds only theoretically.
In fact, for \mbox{$n=2^{18}-1=262143$}, the algorithm of the code returns $1$ solving the equation
\mbox{$x_{19}=1$} on the $\Bigl(2 \cdot 3 \cdot 5 \cdot 7 \cdot 11 \cdot 13 \cdot 17 \cdot 19 \cdot 23
\cdot 29 \cdot 31 \cdot 37 \cdot 41 \cdot 43 \cdot 47 \cdot 53 \cdot 59 \cdot 61 \cdot {67}^2-1\Bigr)$-th
iteration.
\vskip 0.2truecm
\par
To avoid long and and trivial computations, the next code implements a slightly changed function
\mbox{$\xi:\N \to \N$}. The code, which is also stored in \cite{Tyszka2}, computes the system $S$
and the set $V$ of all variables which occur in equations from $S$. If any variable in $V$ has an
index greater than \mbox{${\rm card}(V)$}, then the code continues execution as for the system
\mbox{$S=\{x_1=1\}$}, which is formally defined as the set \mbox{$\{1\}$}.
\begin{verbatim}
input("input the value of n",n):
print(0):
A:=op(ifactor(210*(n+1))):
B:=[A[2*i+1] $i=1..(nops(A)-1)/2]:
S:={}:
V:={}:
for i from 1 to floor(nops(B)/4) do
if B[4*i]=1 then
S:=S union {B[4*i-3]}:
V:=V union {B[4*i-3]}:
end_if:
if B[4*i]=2 then
S:=S union {[B[4*i-3],B[4*i-2],B[4*i-1],"+"]}:
V:=V union {B[4*i-3],B[4*i-2],B[4*i-1]}:
end_if:
if B[4*i]>2 then
S:=S union {[B[4*i-3],B[4*i-2],B[4*i-1],"*"]}:
V:=V union {B[4*i-3],B[4*i-2],B[4*i-1]}:
end_if:
end_for:
if max(V[i] $i=1..nops(V))>nops(V) then S:={1} end_if:
m:=2:
repeat
C:=op(ifactor(m)):
W:=[C[2*i+1]-1 $i=1..(nops(C)-1)/2]:
T:={}:
for i from 1 to nops(W) do
for j from 1 to nops(W) do
for k from 1 to nops(W) do
if W[i]=1 then T:=T union {i} end_if:
if W[i]+W[j]=W[k] then T:=T union {[i,j,k,"+"]} end_if:
if W[i]*W[j]=W[k] then T:=T union {[i,j,k,"*"]} end_if:
end_for:
end_for:
end_for:
m:=m+1:
until S minus T={} end_repeat:
print(max(W[i] $i=1..nops(W))):
\end{verbatim}
\par
The commercial version of {\sl MuPAD} is no longer available as a \mbox{stand-alone} product,
but only as the {\sl Symbolic Math Toolbox} of {\sl MATLAB}. Fortunately, the presented codes
can be executed by {\sl MuPAD Light}, which was and is free, \mbox{see \cite{Tyszka8}}.
\vskip 0.2truecm
\par
Let ${\mathcal P}$ denote a predicate calculus with equality and one binary relation symbol,
and let $\Lambda$ be a computable function that maps $\N$ onto the set of sentences of ${\mathcal P}$.
The following pseudocode in {\sl MuPAD} implements a \mbox{limit-computable} function
\mbox{$\sigma: \N \to \N$} that cannot be bounded by any computable function.
\vskip 0.2truecm
\noindent
{\tt input("input the value of n",n):}\\
{\tt print(0):}\\
{\tt k:=1:}\\
{\tt while} $\Lambda(n)$ $holds$ $in$ $all$ $models$ $of$ $size$ $k$ {\tt do}\\
{\tt k:=k+1:}\\
{\tt end\_while:}\\
{\tt print(k):}
\vskip 0.2truecm
\par
The proof follows from the fact that the set of sentences of ${\mathcal P}$ that are true in all
finite and \mbox{non-empty} models is not recursively enumerable, see \cite[p.~129]{EF},
where it is concluded from Trakhtenbrot's theorem. The author has no idea how to transform the
pseudocode into a correct computer program.
\vskip 0.2truecm
\par
The next theorem, which is now widely known, strengthens Janiczak's result mentioned earlier.
\begin{theorem}\label{the5}
There exists a computable function \mbox{$g:\N \times \N \to \N$} such that the
function \mbox{$f:\N \to \N$} defined by
\begin{displaymath}
f(n)=\left\{
\begin{array}{l}
0, {\rm ~~if~~} \{m \in \N: g(n,m)=0\}=\emptyset\\
{\rm min}\Bigl\{m \in \N: g(n,m)=0\Bigr\}, {\rm ~~if~~} \{m \in \N: g(n,m)=0\} \neq \emptyset
\end{array}
\right.
\end{displaymath}
cannot be bounded by any computable function.
\end{theorem}
\begin{proof}
The last pseudocode and the previous two codes compute appropriate functions.
\end{proof}
\par
The last pseudocode and the previous two codes have metamathematical character as they
syntactically compute mathematical objects (sentences, systems of equations) that correspond
to \mbox{non-negative} integers. The author will try to write a computer code with the
following properties:
\vskip 0.2truecm
\noindent
{\tt 6)} The first instruction takes as input a \mbox{non-negative} integer $n$.
\vskip 0.2truecm
\noindent
{\tt 7)} The execution of the code may terminate or not.
\vskip 0.2truecm
\noindent
{\tt 8)} The second instruction returns $0$.
\vskip 0.2truecm
\noindent
{\tt 9)} The last instruction returns a positive integer.
\vskip 0.2truecm
\noindent
{\tt 10)} All other instructions return nothing.
\vskip 0.2truecm
\noindent
{\tt 11)} The code implements some \mbox{limit-computable} function \mbox{$\mu: \N \to \N$}
that cannot be bounded by any computable function.
\vskip 0.2truecm
\noindent
{\tt 12)} The code does not use any enumeration of mathematical objects which differs from
the canonical enumeration of \mbox{non-negative} integers.
\newpage
A possible algorithm is simple in theory.
Let ${\cal B}$ be any subset of $\N$ which is recursively enumerable but not recursive.
By the Davis-Putnam-Robinson-Matiyasevich theorem, there exists a polynomial
\mbox{$D(x,x_1,\ldots,x_m)$} with integer coefficients such that for each \mbox{$n \in \N$},
\[
n \in {\cal B} \Longleftrightarrow \exists x_1 \ldots x_m \in \N ~~D(n,x_1,\ldots,x_m)=0
\]
The algorithm is presented in the flowchart below.
\begin{center}
\includegraphics[scale=0.88]{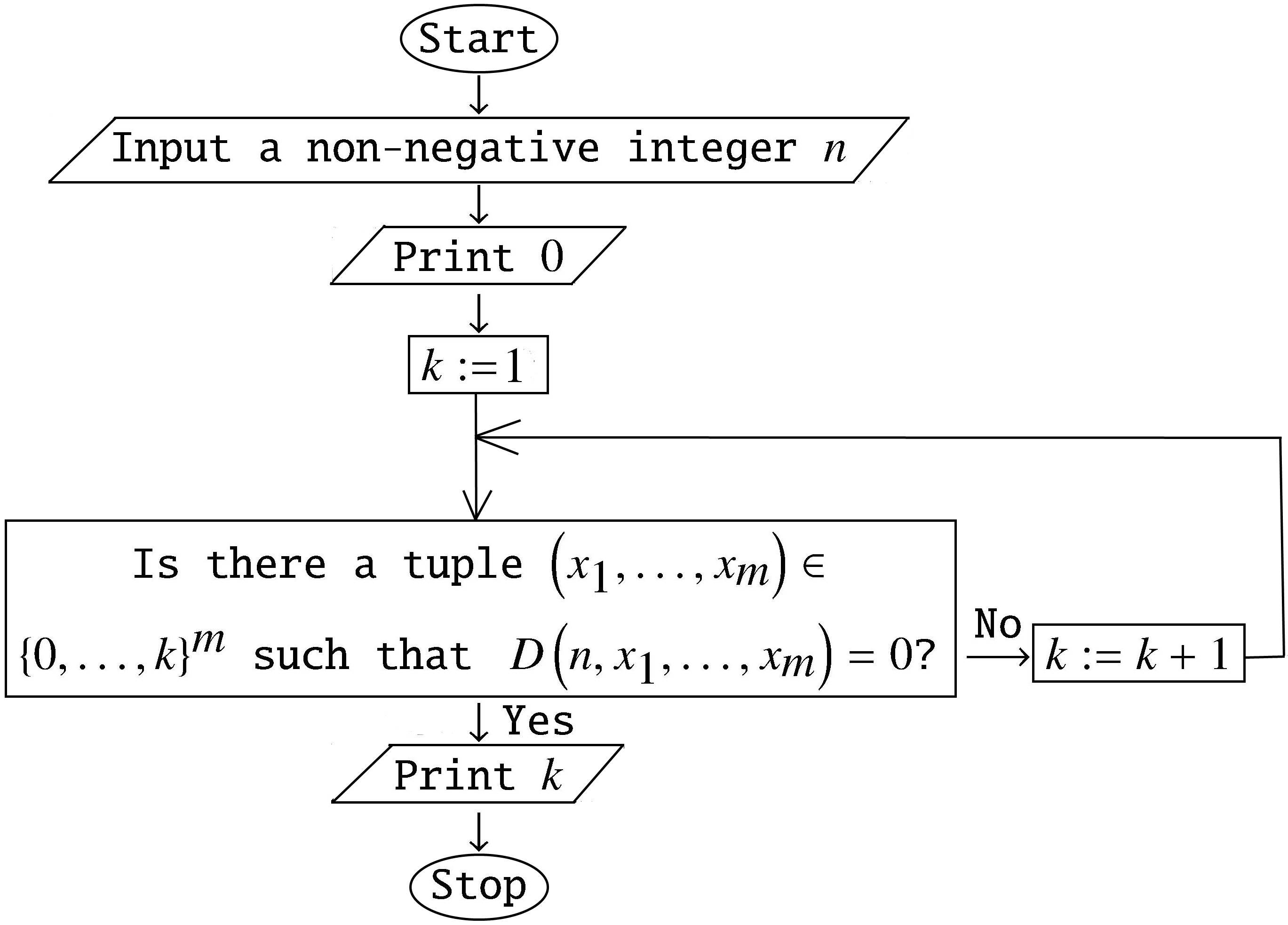}
\end{center}
\par
For each \mbox{non-negative} integer $n$, the implemented function $\mu$ has a \mbox{non-zero} value at $n$ if and
only if \mbox{$n \in {\cal B}$}. Unfortunately, for each recursively enumerable set \mbox{${\cal B} \subseteq \N$}
which is not recursive, any so far constructed polynomial \mbox{$D(x,x_1,\ldots,x_m)$} is complicated and its values
cannot be easily computed.

\noindent
Apoloniusz Tyszka\\
University of Agriculture\\
Faculty of Production and Power Engineering\\
Balicka 116B, 30-149 Krak\'ow, Poland\\
E-mail address: \url{rttyszka@cyf-kr.edu.pl}
\end{sloppypar}

\begin{thebibliography}{17}

\bibitem{EF}
\mbox{H.-D. Ebbinghaus and J. Flum},
\newblock
{\em Finite model theory,}
\newblock
Springer-Verlag, Berlin, 2006.

\bibitem{Janiczak}
\mbox{A. Janiczak},
\newblock
{\em Some remarks on partially recursive functions,} 
\newblock
Colloquium Math. 3 (1954), 37--38.

\bibitem{Matiyasevich}
\mbox{Yu. Matiyasevich},
\newblock
{\em Hilbert's tenth problem,}
\newblock
MIT Press, Cambridge, MA, 1993.

\bibitem{Matiyasevich1}
\mbox{Yu. Matiyasevich},
\newblock
{\em Towards finite-fold Diophantine representations,}
\newblock
Zap. Nauchn. Sem. S.-Peterburg. Otdel. Mat. Inst. Steklov. (POMI) 377 (2010), \mbox{78--90},
\newblock
\url{ftp://ftp.pdmi.ras.ru/pub/publicat/znsl/v377/p078.pdf}.

\bibitem{Murawski}
\mbox{R. Murawski},
\newblock
{\em The contribution of Polish logicians to recursion theory,} in:
\newblock
K. Kijania-Placek and J. Wole\'nski (eds.),
{\em The Lvov-Warsaw School and Contemporary Philosophy,} \mbox{265--282},
\newblock
Kluwer Acad. Publ., Dordrecht, 1998.

\bibitem{Soare}
\mbox{R. I. Soare},
\newblock
{\em Interactive computing and relativized computability,} in:
\newblock
{\em Computability: Turing, G\"odel, Church, and beyond} (eds. \mbox{B. J. Copeland}, \mbox{C. J. Posy}, and \mbox{O. Shagrir}),
\newblock
MIT Press, Cambridge, MA, 2013, \mbox{203--260}.

\bibitem{Tyszka1}
\mbox{A. Tyszka},
\newblock
{\em A code in {\sl MuPAD},}
\newblock
\url{http://www.cyf-kr.edu.pl/~rttyszka/c210.txt}.

\bibitem{Tyszkasf}
\mbox{A. Tyszka},
\newblock
{\em A code in {\sl MuPAD} which performs an infinite computation of the function $\chi$,}
\newblock
\url{http://www.cyf-kr.edu.pl/~rttyszka/sf.txt}.

\bibitem{Tyszka3a}
\mbox{A. Tyszka},
\newblock
{\em A function, whose computability is an open problem, and which dominates all
functions with a n-fold (finite-fold) Diophantine representation,}
\newblock
\mbox{\url{http://arxiv.org/abs/1404.5975}}.

\bibitem{Tyszka1a}
\mbox{A. Tyszka},
\newblock
{\em A \mbox{limit-computable} function which does not have any \mbox{single-fold} Diophantine
representation and whose computability is an open question,}
\newblock
\url{http://arxiv.org/abs/1309.2682}.

\bibitem{Tyszka2}
\mbox{A. Tyszka},
\newblock
{\em A longer code in {\sl MuPAD},}\\
\newblock
\url{http://www.cyf-kr.edu.pl/~rttyszka/changed_xi.txt}.

\bibitem{Tyszka3}
\mbox{A. Tyszka},
\newblock
{\em A {\sl MuPAD} code that computes the values of \mbox{$\varphi(n,l)$},}\\
\newblock
\url{http://www.cyf-kr.edu.pl/~rttyszka/phi.txt}.

\bibitem{Tyszka6}
\mbox{A. Tyszka},
\newblock
{\em An infinite loop in {\sl MuPAD} which returns a finite sequence of \mbox{non-negative} integers,}
\newblock
\url{http://www.cyf-kr.edu.pl/~rttyszka/loop1.txt}.

\bibitem{Tyszka5}
\mbox{A. Tyszka},
\newblock
{\em An infinite loop in {\sl MuPAD} which returns an infinite sequence of \mbox{non-negative} integers,}
\newblock
\url{http://www.cyf-kr.edu.pl/~rttyszka/loop.txt}.

\bibitem{Tyszka7}
\mbox{A. Tyszka},
\newblock
{\em Conjecturally computable functions which unconditionally do not have any \mbox{finite-fold}
Diophantine representation,}
\newblock
Inform. Process. Lett. 113 (2013), \mbox{no. 19–-21}, \mbox{719--722}.

\bibitem{Tyszka8}
\mbox{A. Tyszka},
\newblock
{\em Links to an installation file for MuPAD Light,}\\
\newblock
\mbox{\url{http://www.ts.mah.se/utbild/ma7005/mupad_light_scilab_253.exe}},
\newblock
\mbox{\url{http://caronte.dma.unive.it/info/materiale/mupad_light_scilab_253.exe}},
\newblock
\mbox{\url{http://www.cyf-kr.edu.pl/~rttyszka/mupad_light_scilab_253.exe}},
\newblock
\mbox{\url{http://www.projetos.unijui.edu.br/matematica/amem/mupad/mupad_light_253.exe}},
\newblock
\mbox{\url{http://www.cyf-kr.edu.pl/~rttyszka/mupad_light_253.exe}}.

\bibitem{Tyszka9}
\mbox{A. Tyszka},
\newblock
{\em Small systems of Diophantine equations which have only very large integer solutions,}
\newblock
\mbox{\url{http://arxiv.org/abs/1102.4122}}.
\end{thebibliography}
\end{document}